\documentclass[journal,letter]{IEEEtran}
 \usepackage{graphicx}
 \usepackage{amssymb}
 \usepackage{amsmath}
 \usepackage{algorithm}
 \usepackage{algpseudocode}
 \usepackage{cite}
 \usepackage{stfloats}
 \usepackage{epstopdf}
 \usepackage{mdwlist}
 \usepackage{threeparttable}
 \usepackage[all,2cell]{xy} \UseAllTwocells
 \usepackage{booktabs}
 \usepackage{fancyvrb}
 \usepackage{balance}
 \usepackage{tensor}
 \usepackage{leftidx}
 \usepackage{mathabx}
 \usepackage{color}
 \usepackage{csquotes}
 \usepackage{pdflscape}
 \usepackage{afterpage}
 \usepackage{array}
  \usepackage{cases}
  \usepackage{bm}
  \usepackage{capt-of}
  \usepackage{siunitx}
  \usepackage[caption=false,font=footnotesize]{subfig}

\newtheorem{proposition}{Proposition}

\newtheorem{remark}{Remark}

\AtBeginDocument{}

\newcommand{\Tb}{T_\mathrm{b}}

\newcommand{\argmin}{\mathop{\mathrm{arg\,min}}}

\newcommand{\atx}{a_\mathrm{Tx}}
\newcommand{\arx}{a_\mathrm{Rx}}

\newcommand{\q}{\mathbf{q}}
\newcommand{\bb}{\mathbf{b}}

\makeatletter
\def\endthebibliography{%
	\def\@noitemerr{\@latex@warning{Empty `thebibliography' environment}}%
	\endlist
}
\makeatother
\begin{document}

\title{\LARGE
        Fractionally-Spaced Equalization and Decision Feedback Sequence Detection  for  Diffusive MC
      }

\author{\IEEEauthorblockN{Trang Ngoc Cao,~\IEEEmembership{Student~Member,~IEEE,
		}  Vahid Jamali,~\IEEEmembership{Member,~IEEE,} Nikola Zlatanov,~\IEEEmembership{Member,~IEEE,}\\ Phee Lep Yeoh,~\IEEEmembership{Member,~IEEE,} Jamie Evans,~\IEEEmembership{Senior Member,~IEEE,} and Robert Schober,~\IEEEmembership{Fellow,~IEEE}
}
\thanks{
	This work was supported partially by the Australian Research Council Discovery Projects under Grants DP180101205 and DP190100770 and the German Ministry for Education and Research under the MAMOKO project. 
}%
\thanks{T. N. Cao and J. S. Evans are with the Department of Electrical and Electronic
	Engineering, University of Melbourne, Melbourne, VIC 3010, Australia
	(e-mail: ngocc@student.unimelb.edu.au; jse@unimelb.edu.au).}
\thanks{ V. Jamali and R. Schober are
	with the Institute for Digital Communications, Friedrich-Alexander-
	Universit\"at Erlangen-N\"urnberg, 91058 Erlangen, Germany (e-mail:
	vahid.jamali@fau.de;  robert.schober@fau.de).}
\thanks{N. Zlatanov is with the Department of Electrical and Computer Systems
	Engineering, Monash University, Clayton, VIC 3800, Australia
	(e-mail:	nikola.zlatanov@monash.edu).}
\thanks{P. L. Yeoh is with the School of Electrical and Information Engineering, University of Sydney, Sydney, NSW 2006, Australia
	(e-mail: phee.yeoh@sydney.edu.au).}
}

\maketitle

\begin{abstract}
	In this paper,  we consider diffusive molecular communication (MC) systems affected by signal-dependent diffusive noise, inter-symbol interference, and external  noise. We design linear and nonlinear fractionally-spaced  equalization
	schemes and a detection scheme which combines decision feedback and  sequence detection (DFSD).  In contrast to the symbol-rate equalization schemes in the MC literature, the proposed  equalization and detection schemes exploit multiple samples of the received signal per symbol interval to achieve  lower bit error rates (BERs) than  existing schemes.  The proposed  DFSD scheme achieves a  BER which is very close to that achieved by maximum likelihood sequence detection, but with lower computational complexity.	
\end{abstract}

\section{Introduction}

The diffusive channel has been widely considered for the design of molecular communication (MC) systems due to its simplicity without the need for special infrastructure or external energy \cite{JAW:18:Arxiv}. However, diffusion leads to signal-dependent noise and inter-symbol interference (ISI). 
Various approaches have been proposed in the MC literature to mitigate the impact of ISI including the use of more than one type of molecule \cite{MGM:18:JC,TPK:15:COML},  enzymes to  degrade the  information molecules \cite{NCS:14:INB},   adaptive threshold detection \cite{DH:16:JNB,CLY:18:NB},  reactive signaling \cite{JFS:19:MBSC,FPG:17:GLOBECOM}, matched filtering \cite{JAS:17:CL}, and   equalization \cite{KA:13:JSAC}. Among these techniques,  adaptive threshold detection \cite{DH:16:JNB},   matched filtering \cite{JAS:17:CL}, and equalization \cite{KA:13:JSAC} do not require more than one type of chemical in the system, which is beneficial for keeping the system design complexity low. For  adaptive threshold detection in  \cite{DH:16:JNB} and  equalization in \cite{KA:13:JSAC}, the received signal is  sampled only once per symbol interval. For  adaptive detection in \cite{CLY:18:NB} and the matched filter in \cite{JAS:17:CL}, multiple samples within one symbol interval are used for the detection of that symbol. In this work,  we improve  performance by  designing equalizers and  detectors which exploit multiple samples per symbol interval as well as multiple symbol intervals for  detection of one symbol. In particular, we design linear and nonlinear fractionally-spaced equalizers and a detection scheme combining decision feedback and sequence detection (DFSD).  

Fractionally-spaced equalization and DFSD have been considered  for the mitigation of ISI in conventional communication systems \cite[Chapter 9]{ Ben:15:CL,Lee:77:TCOM, Due:89:TC,Mar:07:SP,PRO:01:Book}. However, the  schemes in \cite[Chapter 9]{ Ben:15:CL,Lee:77:TCOM,Due:89:TC,Mar:07:SP,PRO:01:Book} are designed for independent additive noise whereas MC is also affected by signal-dependent diffusion noise. Moreover, to the best of the authors' knowledge, the performance of fractionally-spaced equalizers and   DFSD  has not yet been investigated for MC.

\section{System Model}

We consider a diffusive  MC system in an unbounded three dimensional  environment comprising  a  spherical transparent transmitter, a  spherical passive receiver,  and information molecules. Let $\atx$ and $\arx$ denote the radius of the transmitter and the receiver, respectively.  The  distance between the transmitter and the receiver is denoted by $r$. %In this work, we are interested in both passive receiver and absorbing receiver.
We  assume that the molecular movement is caused by Brownian motion  with diffusion coefficient $D$ and a uniform flow.  Let $v_1$ and $v_2$ denote the parallel and perpendicular components of the uniform flow  from the transmitter to the receiver, respectively.
% Let $r_0$ denoted $r(t=0)$. 

The transmitter employs on-off keying modulation to convey information to the receiver. 
%We will develop a framework to be applied for both types of the receivers, i.e.,  passive receiver and absorbing receiver. 
A sequence of $K$ symbols, with one  bit  per symbol, is transmitted. As shown in Fig.~\ref{fig:1}, the transmitted symbol  and the detected symbol at the receiver are denoted by $s_k \in \left\{0,1\right\}$ and $\hat{s}_k$,  $k=1,2,\dots, K$, respectively. At the beginning of the $k$-th symbol interval, at time $t_{k,0}$,  the transmitter releases $A$ molecules  to transmit bit $``1"$ or is silent to transmit bit $``0"$. We assume that  the probability of transmitting bit ${s}_k$ is $\Pr\left({s}_k=0\right)=\Pr\left({s}_k=1\right)=1/2$.

  		\begin{figure}[!t]
  			\centering
  			\includegraphics[scale=0.65, trim={0 0.6cm 0 0.6cm},clip]{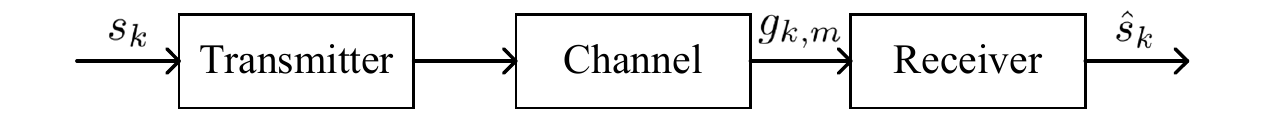}
  			\caption{
  				System model with input and output signals. 
  			}
  			\label{fig:1}
  			%		\end{minipage}
  		\end{figure}

 Let $\Tb$ denote the length of the symbol interval. We assume that the receiver collects $M$ samples of the received signal during each symbol interval. The $m$-th sample of the $k$-th symbol is denoted by $g_{k,m}, m=1,2,\dots,M$. Let $t_{k,m}=(k-1)\Tb+m \Delta t$ denote the $m$-th sampling time in the $k$-th symbol interval, where $ \Delta t\leq \Tb,$ is the sampling interval. For a passive receiver, $g_{k,m}$ denotes the number of molecules observed in the volume of the receiver at time $t_{k,m}$. 
 %For absorbing receiver,  $g_{k,m}$ is the number of molecules absorbed by $\Rx$ during $t=[t_{k,m-1},t_{k,m}]$.
 We  assume
that the $k$-th symbol is affected by
significant ISI  from the  $L - 1$ previous symbols. The combined impact of the ISI originating from symbols transmitted before the $(k-L+1)$-th symbol and the noise from external sources is modeled by an additive interference signal \cite{NCS:14:JSAC}. The additive interference signal has a constant expected value, denoted by $\eta$.
It is shown in \cite{NCS:14:JNB} that $g_{k,m}$ follows the Binomial distribution and that it can be well
 approximated as a Poisson or Gaussian random variable, where for typical MC applications the Poisson distribution is a more accurate approximation \cite{JAW:18:Arxiv}. In this work, we use the Poisson distribution to model  the received samples, and thus, we have   
\begin{align} \label{eq:1}
 g_{k,m} &\sim  \mathcal{P}\left(\sum_{l=1}^L s_{k-l+1} h_{l,m}+\eta\right),
\end{align}
where  $h_{l,m}$ is the expected number of molecules received at the receiver at time $t_{k,m}$ due to the release of
$A$ molecules at the transmitter at time  $t=(k-l)\Tb$. In other words, $h_{l,m}$ is received after a period, denoted by $T_{l,m}$, which is equal to $T_{l,m}=t_{k,m}-(k-l)\Tb=(k-1)\Tb+m \Delta t-(k-l)\Tb=(l-1)\Tb+m\Delta t$.
For a passive receiver, $h_{l,m}$ is given  by \cite{JAS:17:CL}
\begin{align} \label{eq:2}
h_{l,m}=\frac{A V}{(4\pi D T_{l,m})^{3/2}}\exp\left(\frac{-\left(r-v_1 T_{l,m}\right)^2+\left(v_2 T_{l,m}\right)^2}{4D T_{l,m}}\right),
\end{align}
where $V$ is the volume of the receiver.
 We assume that the received numbers of molecules from different transmissions are independent and $g_{k,m}$ is independent $\forall k,m$ \cite{NCS:14:JNB,JAW:18:Arxiv}.

\begin{remark}
	 The expression in \eqref{eq:1} shows that the MC system is affected by signal-dependent noise. Eq.~\eqref{eq:1} holds for both   passive receivers and  absorbing receivers and only the values of $h_{l,m}$ depend on the type of receiver. Therefore, the equalization and detection schemes proposed in the following sections can  be applied in systems employing  passive receivers or  absorbing receivers. Here, we adopt a passive receiver for the numerical results shown in Section~\ref{sec:5}.
\end{remark}

In the following, we design fractionally-spaced equalization and detection schemes for  ISI mitigation at the receiver.

\section{Fractionally-Spaced Equalization for MC} \label{sec:3}

In this section, we propose  linear and nonlinear fractionally-spaced equalizers for MC systems.  Due to the ISI, the symbol to-be-detected  is influenced by  previously transmitted symbols and the symbol to-be-detected  also influences  the following symbols. To exploit the resulting dependencies for  detection of the considered symbol, we use the received signals corresponding to the previous,  to-be-detected, and  following symbols for the design of the equalizers. 
\subsection{Linear Fractionally-Spaced  Equalization} \label{sub3:1}
\begin{figure*}[!tbp]
	\vspace{-0.6cm}
	\centering
		\hspace{-1.4cm}
	\subfloat[]{
		\includegraphics[scale=0.5, trim={0 0.6cm 0 0.6cm},clip]{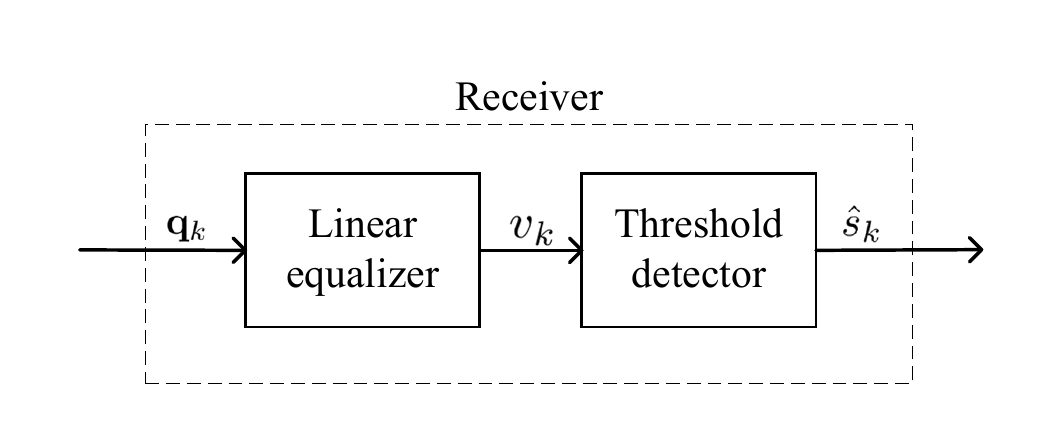}
		\label{fig:2}
	}
%	\hspace{-0.8cm}
	\subfloat[]{
		\includegraphics[scale=0.5, trim={0 0.6cm 0 0.6cm},clip]{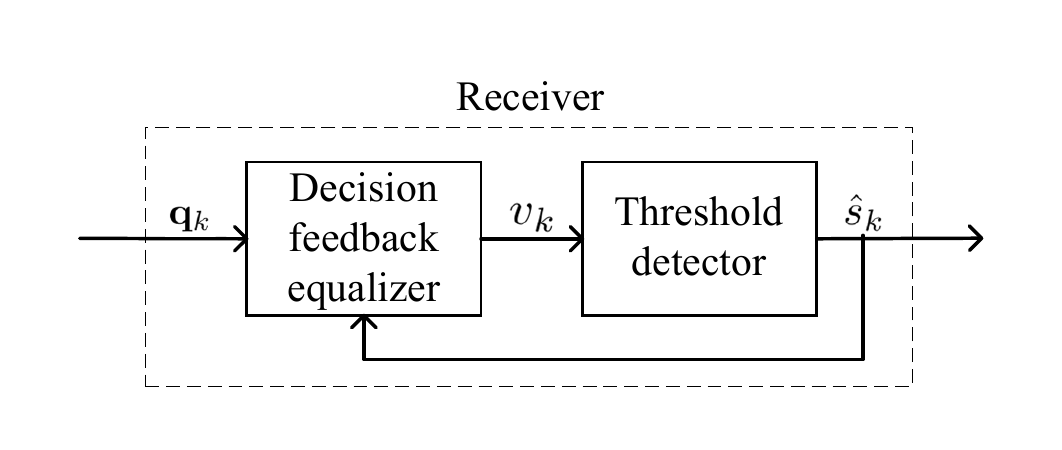}
		\label{fig:3}
	}
%	\hspace{-0.85cm}
	\subfloat[]{
		\includegraphics[scale=0.5, trim={0.1cm 0.6cm 0 0.6cm},clip]{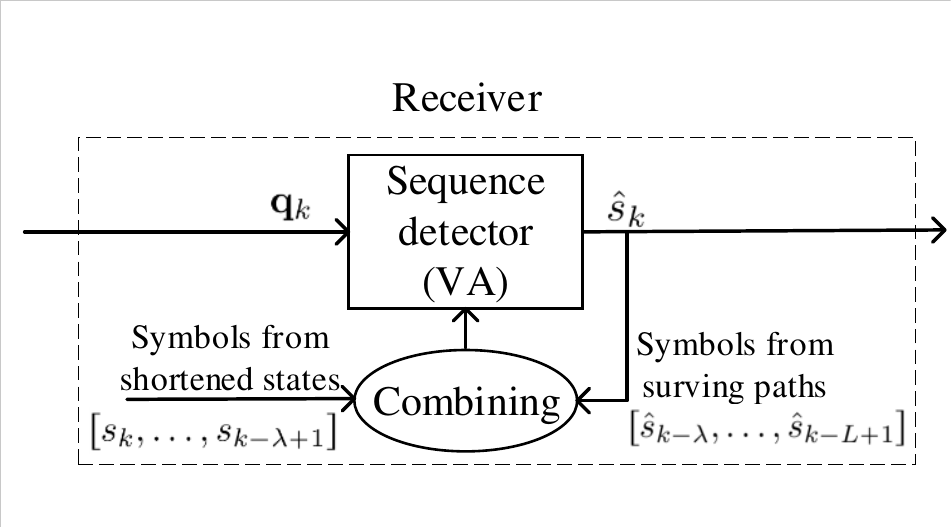}
		\label{fig:4}
	}
	\hspace{-1.3cm}
	\caption{Diagram of a receiver with
		(a) a linear equalizer and a threshold detector, (b) a decision feedback equalizer and a threshold detector, or (c) decision feedback sequence detection. 
		%\vspace*{-5mm}
	}
	\label{fig:dia}
	%\vspace{-0.5cm}
\end{figure*}

In this subsection, we consider a linear fractionally-spaced equalizer. 
Let $\q_k$ and  $v_k$ denote the input vector and the output of the equalizer for the detection of the $k$-th symbol, respectively, see Fig.~\ref{fig:2}. Vector
$\q_k$ is comprised of the elements $g_{k,m}$ spanning $T$ symbol intervals before and after the current symbol, respectively. Specifically, the $i$-th element of $\q_k$, $i=\{1,\cdots, (2T+1)M\}$, is defined as
\begin{align}
q_k[i]\overset{\Delta}{=}\begin{cases}
g_{k-T-1+\lfloor i/M\rfloor,|i|_M} &\text{if } |i|_M\neq 0,\\
g_{k-T-1+\lfloor i/M\rfloor,M} &\text{if } |i|_M=0,
\end{cases}
\end{align}
where  $\lfloor\cdot\rfloor$ is the floor operation and $|\cdot|_M$ is the modulo $M$ operation.
The output $v_k$ is an affine function of the input sequence $\q_k$  given by
\begin{align} \label{eq:2}
v_k&=\sum_{i=1}^{(2T+1)M}b[i] q_k[i]+b^\mathrm{c}=\mathbf{b}^\mathsf{T}\mathbf{q}_k+b^\mathrm{c},
\end{align}
where   $\left\{b[i]\right\}$ and $b^\mathrm{c}$ are the $(2T+1)M+1$ tap weight coefficients of the equalizer, $b[i]$ is the $i$-th element of vector $\mathbf{b}$, and $\mathbf{b}^\mathsf{T}$ is the transpose of $\mathbf{b}$. Coefficient $b^\mathrm{c}$ is usually equal to zero for conventional wireless communication applications due to the zero-mean input and output signals. However, this may not be the case for  MC applications.
 Ideally,  $\mathbf{b}$ and $b^\mathrm{c}$ should be chosen to minimize the BER. However, the  expression for the BER is complicated and thus  the corresponding optimal tap weight coefficients  cannot be obtained in closed form. Therefore, we optimize the tap weight coefficients in terms of the minimum mean squared error (MMSE) between the equalizer's output and the transmitted symbol, which is expected to also result in a low BER.

Let $\varepsilon_k=s_k-v_k$ denote the error between the equalizer's output and the transmitted symbol. The linear fractionally-spaced equalizer that achieves the MMSE between the equalizer's output and the transmitted symbol is given as follows
\begin{align} \label{eq:13}
\left\{\mathbf{b}_{\mathrm{opt}},b^\mathrm{c}_{\mathrm{opt}}\right\}&=\argmin_{\left\{\mathbf{b},b^\mathrm{c}\right\}}\mathbf{E}\left\{\varepsilon_k^2\right\}=\argmin_{\left\{\mathbf{b},b^\mathrm{c}\right\}}\mathbf{E}\left\{\left(s_k-v_k\right)^2\right\}\nonumber\\
&=\argmin_{\left\{\mathbf{b},b^\mathrm{c}\right\}}\mathbf{E}\left\{\left(s_k-\mathbf{b}^\mathsf{T}\q_k-b^\mathrm{c}\right)^2\right\},
\end{align}
where $x_{\mathrm{opt}}$ and $\mathbf{E}\{\cdot\}$ denote the optimal value of $x$ and  expectation, respectively. The solution of \eqref{eq:13} is given in the following proposition. To this end, we define $\mathbf{H}$ and $\bm{\Gamma}$ as follows. The element in the $j$-th row ($1\leq j \leq 2T+1$) and the $m$-th column ($1 \leq m \leq M$) of matrix $\mathbf{H}$ is given  by
\begin{align} \label{eq:4}
\begin{cases}
\mathbf{H}[j,m]=h_{j-T,m}, \hspace{1cm} &\text{if } T+1\leq j\leq 2T+1,\\
\mathbf{H}[j,m]=0,& \text{otherwise.}
\end{cases}
\end{align}
The element in the  $i$-th row and the $i'$-th column of matrix $\bm{\Gamma}$ is given by 
	\begin{align} \label{eq:10}
\bm{\Gamma}[i,i']=
\begin{cases}
\frac{1}{2}\sum_{l=1}^L h_{l,|i|_M}+\frac{1}{4}\sum_{l=1}^Lh_{l,|i|_M}^2+\eta  &\text{if } i=i'\\
\frac{1}{4}\sum_{l=1}^{L-\lfloor (i-i')/M\rfloor} h_{l+\lfloor (i-i')/M\rfloor,|i|_M} h_{l,|i'|_M} &\text{if } i>i'\\
\frac{1}{4}\sum_{l=1}^{L-\lfloor (i'-i)/M\rfloor}h_{l,|i|_M} h_{l+\lfloor (i'-i)/M\rfloor,|i'|_M} &\text{if } i<i'.\\
\end{cases}
\end{align}

\begin{proposition} \label{prop:1}
	The optimal coefficients of the linear fractionally-spaced MMSE equalizer  are given by 
	\begin{align} \label{eq:9}
	\begin{cases}
	\mathbf{b}_\mathrm{opt}=\bm{\Gamma}^{-1}\bm{\xi},\\
	{b^\mathrm{c}}_\mathrm{opt}=\frac{1}{2}-\bm{\xi}^\mathsf{T}\bm{\Gamma}^{-1}\mathbf{E}\left\{\q_k\right\},
	\end{cases}
	\end{align}
	where    $\bm{\xi}=\frac{1}{4} \mathsf{vec}\left(\mathbf{H}^\mathsf{T}\right)$ and $\mathsf{vec}\left\{\mathbf{H}\right\}$  denote the vectorization  of  matrix $\mathbf{H}$.  
	The $i$-th  element of vector $\mathbf{E}\left\{\q_k\right\}$, i.e., the expectation of $\q_k$ over all possible values of transmitted sequences,  is given by
	\begin{align} \label{eq:11}
	\mathbf{E}\left\{\q_k[i]\right\}=\frac{1}{2}\sum_{l=1}^L h_{l,|i|_M}+\eta.
	\end{align}
\end{proposition}
\begin{proof}
	To obtain the optimal values of $\mathbf{b}$ and $b^\mathrm{c}$, we use the framework for designing a linear MMSE estimator for a non-zero mean variable in  \cite{Kay:93:Book}. Variables $\left\{\mathbf{b},b^\mathrm{c}\right\}_{\mathrm{opt}}$ are obtained by setting the partial derivatives of $\mathbf{E}\left\{\left(s_k-\mathbf{b}^\mathsf{T}\q_k-b^\mathrm{c}\right)^2\right\}$ with respect to $b^\mathrm{c}$ and $\mathbf{b}$ to  zero, respectively. Note that $\mathbf{E}\left\{s_k\right\}=1/2$ as  $\Pr\left({s}_k=0\right)=\Pr\left({s}_k=1\right)=1/2$, whereas $\bm{\xi}$ and $\bm{\Gamma}$ are obtained from $\mathbf{E}\left\{\left(s_k-\mathbf{E}\left\{s_k\right\}\right)\left(\q_k-\mathbf{E}\left\{\q_k\right\}\right)\right\}$ and $\mathbf{E}\left\{\left(\q_k-\mathbf{E}\left\{\q_k\right\}\right)^2\right\}$, respectively, by using the independence of the received signal samples and the independence of the $s_k$.
\end{proof}

%%%%%%%%%%%%%%%%%%%%%%%%%%%%%%%%%%%%%%%%%%%%%%%%%%%%%%%%%%%%%%%%%%%%%%%%%%%%%%%%%%%%%%%%%%%%%%
%\setcounter{eqnback1}{\value{equation}} \setcounter{equation}{6}
%\begin{figure*} [!t]
%	
%
%	%
%	\begin{align} \label{eq:10}
%	\bm{\Gamma}[i,i']=
%	\begin{cases}
%	\frac{1}{2}\sum_{l=1}^L h_{l,|i|_M}+\frac{1}{4}\sum_{l=1}^Lh_{l,|i|_M}^2+\eta  &\text{if } i=i'\\
%	\frac{1}{4}\sum_{l=1}^{L-\lfloor (i-i')/M\rfloor} h_{l+\lfloor (i-i')/M\rfloor,|i|_M} h_{l,|i'|_M} &\text{if } i>i'\\
%	\frac{1}{4}\sum_{l=1}^{L-\lfloor (i'-i)/M\rfloor}h_{l,|i|_M} h_{l+\lfloor (i'-i)/M\rfloor,|i'|_M} &\text{if } i<i'\\
%	\end{cases}
%	\end{align}
%		\hrulefill
%	\setcounter{eqncnt1}{\value{equation}}
%	\setcounter{equation}{\value{eqnback1}}
%\end{figure*}
%%%%%%%%%%%%%%%%%%%%%%%%%%%%%%%%%%%%%%%%%%%%%%%%%%%%%%%%%%%%%%%%%%%%%%%%%%%%%%%%%%%%%%%%%%%%%%%

\subsection{Decision-Feedback Equalization} \label{sub3:3}

In this subsection, we  consider a nonlinear fractionally-spaced equalizer. In particular, we design a decision-feedback equalizer (DFE)    employing two filters, a feedforward filter and a feedback filter \cite{PRO:01:Book}. The input to the feedforward filter, denoted by $\q_k$, is the vector of received samples from the $k$-th symbol and the $L_1$ following symbols. The input to the feedback filter, denoted by $\hat{\mathbf{s}}_{k-1,k-L_2}=[\hat{s}_{k-1},\cdots,\hat{s}_{k-L_2}]^\mathsf{T}$, is a vector containing the $L_2$   symbols detected prior to the $k$-th symbol,  see Fig.~\ref{fig:3}.  For a simple design, we use linear filters for the feedforward and feedback filters.
  The output of the fractionally-spaced DFE  is given by
\begin{align} \label{eq:30}
v_k&=\sum_{i=1}^{(L_1+1) M} b[i] q_k[i]-\sum_{\tau=1}^{L_2}a[\tau]\hat{s}_{k-\tau}+b^\mathrm{c}\\\nonumber
&=\mathbf{b}^\mathsf{T}\mathbf{q}_k-\mathbf{a}^\mathsf{T}\hat{\mathbf{s}}_{k-1,k-L_2}+b^\mathrm{c},
\end{align}
where $b[i]$ and $a[\tau]$ are the coefficients of the feedforward and feedback filters respectively, and  $b^\mathrm{c}$ is a constant coefficient. In \eqref{eq:30}, $b[i]$ and $q_k[i]$  are the $i$-th elements of vectors  $\bb$ and $\q_k$, respectively, $a[\tau]$ and $\hat{s}_{k-\tau}$ are the $\tau$-th elements of vectors $\mathbf{a}$ and $\hat{\mathbf{s}}_{k-1,k-L_2}$, respectively. 
 Here, $q_k[i]$ is given by 
\begin{align}
q_k[i]\overset{\Delta}{=}\begin{cases}
g_{k+\lfloor i/M\rfloor,|i|_M} &\text{if } |i|_M\neq 0,\\
g_{k+\lfloor i/M\rfloor,M} &\text{if } |i|_M=0.
\end{cases}
\end{align}

Due to the feedback of previous decisions in the DFE design, a closed-form
expression of the BER cannot be obtained and thus the DFE
filter cannot be optimized for minimization of the BER. Therefore, we 
 optimize the coefficients of the feedforward and feedback filters in terms of MMSE with the assumption that previous decisions are correct as follows 
\begin{align} \label{eq:31}
\left\{\mathbf{b},\mathbf{a},b^\mathrm{c}\right\}_{\mathrm{opt}}=\argmin_{\mathbf{b},\mathbf{a},b^\mathrm{c}}\mathbf{E}\left\{\varepsilon_k^2\right\}=\argmin_{\mathbf{b},\mathbf{a},b^\mathrm{c}}\mathbf{E}\left\{\left(s_k-v_k\right)^2\right\}.
\end{align}

The following proposition provides  the solution of \eqref{eq:31}.

\begin{proposition} \label{prop:4}
	The optimal coefficients, $\mathbf{b}_{\mathrm{opt}}$, $\mathbf{a}_{\mathrm{opt}}$, and $b_{\mathrm{opt}}^\mathrm{c}$, of the fractionally-spaced DFE in terms of MMSE, under the assumption that previous decisions are correct, are respectively given as follows.
    The vector $\mathbf{b}_{\mathrm{opt}}$  can be found from
	\begin{align} \label{eq:33}
	\left(\bm{\bar{\Gamma}}-4\mathbf{H}_{\mathbf{s}\q}\right)\mathbf{b}_{\mathrm{opt}}=\bar{\mathbf{h}},
	\end{align}
	where $\bm{\bar{\Gamma}}$ is identical to $\bm{\Gamma}$  in \eqref{eq:10} with $1\leq i \leq (L_1+1)M$ and $1\leq i' \leq (L_1+1)M$.  The element  in the $(j'M+m')$-th row and the $(jM+m)$-th column of $\mathbf{H}_{\mathbf{s}\q}$ in \eqref{eq:33} is given by
	\begin{align} \label{eq:34}
	\mathbf{H}_{\mathbf{s}\q}[j'M+m',jM+m]=\sum_{\tau=1}^{L_2} H_{{j',m'}}^{(\tau)} H_{{j,m}}^{(\tau)},
	\end{align}
	where ${H}_{{j',m'}}^{(\tau)}$ is given by
	\begin{align} \label{eq:35}
	H_{{j',m'}}^{(\tau)}&=\frac{1}{4} h_{\tau+j'+1,m'}.
	\end{align}
	$\bar{\mathbf{h}}=\mathrm{vec}\left(\bm{\zeta}^\mathrm{T}\right)$, where the element in the $j'$-th row and the $m'$-th column of $\bm{\zeta}$ is  $\zeta[j',m']=H_{{j',m'}}^{(0)}$.
	
	The $\tau$-th element of vector $\mathbf{a}_{\mathrm{opt}}$ is given by
	\begin{align} \label{eq:36}
	{a}_{\mathrm{opt}}[\tau]=4\sum_{i=1}^{(L_1+1)M}b_{\mathrm{opt}}[i]H_{{\lfloor i/M\rfloor,m}}^{(\tau)}.
	\end{align}
	Finally, $b_{\mathrm{opt}}^\mathrm{c}$ is given by
		\begin{align} \label{eq:37}
		b_{\mathrm{opt}}^\mathrm{c}&=\sum_{i=1}^{(L_1+1) M} b_\mathrm{opt}[i] \mathbf{E}\left\{\mathbf{q}_k[i]\right\}-\frac{1}{2}\sum_{\tau=1}^{L_2}a_\mathrm{opt}[\tau]+\frac{1}{2}.
		\end{align}
	
\end{proposition}
\begin{proof}
The coefficients of the DFE	whose output is given by \eqref{eq:30} are derived by setting the  partial derivatives of $\mathbf{E}\left\{\left(s_k-v_k\right)^2\right\}$ with respect to $\mathbf{a},\mathbf{b}$, and $b^\mathrm{c}$  equal to zero, respectively. %Due to the  space constraint, we do not provide the full proof here. 
A detailed framework can be found in \cite{GBS:71:COMT}. The mean and variance of the binary $s_k$ are given by $\mathbf{E}\left\{s_k\right\}=\frac{1}{2}$ and $\mathbf{Var}\left\{s_k\right\}=\frac{1}{4}$, respectively, due to the assumption that bits $``0"$ and $``1"$ are transmitted independently and with equal probabilities.
	\end{proof}

\section{Detection Schemes}

In this section, we first review the maximum likelihood sequence detector (MLSD), which is the optimal detection scheme, and discuss the simple symbol-by-symbol threshold detector.  MLSD is used as a benchmark. The threshold detector is used in combination with  the  equalizers proposed in Section~\ref{sec:3} and has a lower computational complexity compared to  MLSD. We then propose a detector combining decision feedback and sequence detection to achieve a better performance than threshold detection.

\subsection{Maximum Likelihood Sequence Detection} \label{sub4:1}

Let $p(\mathbf{q}|\mathbf{s}_{K,1})$ be the joint probability density function (PDF) of the received signal vector  $\mathbf{q}$
conditioned on the transmitted symbol sequence  $\mathbf{s}_{K,1}=[s_K, s_{K-1}, \cdots, s_1]$. Here,  $\mathbf{q}=\mathsf{vec}\left(\bm{\Omega}^\mathsf{T}\right)$  and the element in the $k$-th row and $m$-th column of matrix $\bm{\Omega}$ is equal to $g_{k,m}$. The MLSD is given by \cite{PRO:01:Book}
\begin{align} \label{eq:18}
\hat{\mathbf{s}}_{K,1}&=\underset{\mathbf{s}_{K,1}}{\arg} \max p(\mathbf{q}|\mathbf{s}_{K,1})\\ \nonumber
&\overset{(a)}{=}\underset{\mathbf{s}_{K,1}}{\arg} \max \prod_{k=1}^K \prod_{m=1}^M p_{g_{k,m}|\mathbf{s}_{K,1}}(g_{k,m}),\\ \nonumber
&=\underset{\mathbf{s}_{K,1}}{\arg} \max \sum_{k=1}^K \sum_{m=1}^M \ln p_{g_{k,m}|\mathbf{s}_{K,1}}(g_{k,m}),
\end{align}
where $\ln$ is the natural logarithm, $\hat{\mathbf{s}}_{K,1}=[\hat{s}_K, \hat{s}_{K-1}, \cdots, \hat{s}_1]$ contains the detected symbols corresponding to sequence $\mathbf{s}_{K,1}$, $p_{g_{k,m}|\mathbf{s}_{K,1}}(g_{k,m})$ is the PDF of the received signal $g_{k,m}$ conditioned on the transmitted symbol sequence  $\mathbf{s}_{K,1}$ and can be obtained from \eqref{eq:1}, and $(a)$ is due to the mutual independence of the $g_{k,m}$.

%Instead of waiting to receive the entire vector $\mathbf{q}$ to detect $\mathbf{s}_{1,K}$, we use the Viterbi algorithm (VA)  to detect $s_k$ after a delay of $T$ symbols, with $T<K$, after reception of $s_{k+T}$   \cite{PRO:01:Book,NCS:14:JNB}. 
For a channel with a memory of $L-1$ symbols, the MLSD in \eqref{eq:18} can be performed recursively using the Viterbi algorithm (VA) \cite{PRO:01:Book,NCS:14:JNB}. The VA has $2^{L-1}$ states and each state  at time $k-1$ is defined by $\mathbf{s}_{k-1,k-L+1}=[s_{k-1}, s_{k-2}, \cdots, s_{k-L+1}]$.  The transition from state $\mathbf{s}_{k-1,k-L+1}$ to $\mathbf{s}_{k,k-L+2}$ depends on $s_k$. Due to \eqref{eq:18}, the accumulated path metrics $PM_k(\mathbf{s}_{k,k-L+2})$ of  state $\mathbf{s}_{k,k-L+2}$  are given by
\begin{align} \label{eq:19}
PM_k\left(\mathbf{s}_{k,k-L+2}\right)&=PM_{k-1}\left(\mathbf{s}_{k-1,k-L+1}\right)\\\nonumber
&\quad+\sum_{m=1}^M \ln p_{g_{k,m}|s_k, \dots, s_{k-L+1}}(g_{k,m}).
\end{align} 
At each state $\mathbf{s}_{k,k-L+2}$, the surviving path is selected by maximizing the path metrics $PM_k\left(\mathbf{s}_{k,k-L+2}\right)$ with respect to  $s_k$, see \cite{PRO:01:Book} for details.

\subsection{Symbol-by-Symbol Threshold Detection}\label{sub4:2}

The symbol-by-symbol threshold detection is given by
\begin{align}\label{eq:42}
\hat{s}_k=\begin{cases}
1, \text{ if } v_k\geq \gamma  \\
0, \text{ otherwise},
\end{cases}
\end{align}
where $\gamma$ is the detection threshold.
Considering \eqref{eq:13} and \eqref{eq:31}, we  choose $\gamma=\mathbf{E}\left\{s_k\right\}=\frac{1}{2}$ for both proposed equalizers.

 As mentioned above, the BER of the system with a nonlinear equalizer cannot be given in closed-form. Hence, we will evaluate it numerically. Here, we  derive the BER of the system with the linear equalizer.

 Due to the channel memory of $L$ symbols  and the input of  the equalizer including $2T+1$ symbols,     the sequence that affects the detection of $s_k$ is  $\mathbf{s}_{k+T,k-L-T+1}=[s_{k+T},\cdots, s_k,\cdots, s_{k-T}, \cdots, s_{k-L-T+1}]$. Then, $s_k$, which  we want to detect from $v_k$, is the $(T+1)$-th element of $\mathbf{s}_{k+T,k-L-T+1}$. 
 From \eqref{eq:42}, the BER is obtained as
 \begin{align} \label{eq:43}
 P_e&=\hspace{-1mm}\sum_{\forall \mathbf{s}_{k+T,k-L-T+1}}\hspace{-1mm}P_e^c[s_k|\mathbf{s}_{k+T,k-L-T+1}] \Pr\left\{\mathbf{s}_{k+T,k-L-T+1}\right\}\nonumber\\
 &=\frac{1}{2^{2T+L}}\sum_{\forall \mathbf{s}_{k+T,k-L-T+1}}P_e^c[s_k|\mathbf{s}_{k+T,k-L-T+1}],
 \end{align}
 where 
 \begin{align} \label{eq:44}
 &P_e^c[s_k|\mathbf{s}_{k+T,k-L-T+1}]\\\nonumber
 &\hspace{1.1cm}=\Pr\left\{v_k<\gamma|\mathbf{s}_{k+T,k-L-T+1},s_k=1\right\}\Pr\left\{s_k=1\right\}\\\nonumber
 &\hspace{1.1cm}\quad+\Pr\left\{v_k\geq\gamma|\mathbf{s}_{k+T,k-L-T+1},s_k=0\right\}\Pr\left\{s_k=0\right\}
 \end{align}
 and $\Pr\left\{s_k=0\right\}=\Pr\left\{s_k=1\right\}=1/2$.
 In order to obtain the BER, we require $\Pr\left\{v_k<\gamma|\mathbf{s}_{k+T,k-L-T+1},s_k=1\right\}$ and $\Pr\left\{v_k\geq\gamma|\mathbf{s}_{k+T,k-L-T+1},s_k=0\right\}$.
 Since the output of the equalizer is a weighted sum of Poisson random variables,  its PDF is not available in closed form. Therefore,  we  approximate the Poisson distribution of the input of the linear equalizer in \eqref{eq:1} by a Gaussian distribution. Then, since a linear combination of Gaussian random variables is also Gaussian distributed,  the output of the linear equalizer for a particular  sequence of  information symbols, i.e., $\mathbf{s}_{k+T,k-L-T+1}$, approximately follows the following Gaussian distribution 
 \begin{align} \label{eq:37}
 v_k\sim\mathcal{N}\left(\mu_\alpha\left(\mathbf{s}_{k+T,k-L-T+1}\right),\sigma_\alpha^2\left(\mathbf{s}_{k+T,k-L-T+1}\right)\right),\nonumber \\\hfill \alpha \in \left\{0,1\right\},
 \end{align}
 where $\mu_\alpha\left(\mathbf{s}_{k+T,k-L-T+1}\right)$ and $\sigma_\alpha^2\left(\mathbf{s}_{k+T,k-L-T+1}\right)$ are given by
 \begin{align}
 \begin{cases}
 \mu_0\left(\mathbf{s}_{k+T,k-L-T+1}\right)=\frac{1}{2}+\bm{\xi}^\mathsf{T}\bm{\Gamma}^{-1}\left(\mathsf{vec}(\bm{\nu}_{0}^\mathsf{T})-\mathbf{E}\left\{\q_k\right\}\right),
 \\
 \mu_1\left(\mathbf{s}_{k+T,k-L-T+1}\right)=\frac{1}{2}+\bm{\xi}^\mathsf{T}\bm{\Gamma}^{-1}\left(\mathsf{vec}(\bm{\nu}_{1}^\mathsf{T})-\mathbf{E}\left\{\q_k\right\}\right), 
 \\
 \sigma_0^2\left(\mathbf{s}_{k+T,k-L-T+1}\right)=\bm{\xi}^\mathsf{T}\bm{\Gamma}^{-1}\left(\mathsf{diag}\left\{\mathsf{vec}(\bm{\nu}_{0}^\mathsf{T})\right\}\right)\left(\bm{\Gamma}^{-1}\right)^\mathsf{T}\bm{\xi}, 
 \\
 \sigma_1^2\left(\mathbf{s}_{k+T,k-L-T+1}\right)=\bm{\xi}^\mathsf{T}\bm{\Gamma}^{-1}\left(\mathsf{diag}\left\{\mathsf{vec}(\bm{\nu}_{1}^\mathsf{T})\right\}\right)\left(\bm{\Gamma}^{-1}\right)^\mathsf{T}\bm{\xi}, 
 \end{cases}
 \end{align}
 where the element in the $j$-th row, $j=1,\dots, 2T+1$, and the $m$-th column of $\bm{\nu}_{\alpha}$, with $\alpha=s_k=\left\{0,1\right\}$, is given by
 \begin{align} \label{eq:38}
 \bm{\nu}_{\alpha}[j,m]=
 \sum_{l=1}^L  h_{l,m}s_{k-T+j-l}+\eta.
 \end{align}

 Hence, we have
 \begin{align} \label{eq:25}
 &\Pr\left\{v_k<\gamma|\mathbf{s}_{k+T,k-L-T+1},s_k=1\right\}\\\nonumber
 &\hspace{3cm}=1-\mathbf{Q}\left(\frac{\gamma-\mu_1\left(\mathbf{s}_{k+T,k-L-T+1}\right)}{\sigma_1\left(\mathbf{s}_{k+T,k-L-T+1}\right)}\right)
 \end{align}
 and
 \begin{align}\label{eq:26}
 &\Pr\left\{v_k\geq\gamma|\mathbf{s}_{k+T,k-L-T+1},s_k=0\right\}\\\nonumber
 &\hspace{3cm}=\mathbf{Q}\left(\frac{\gamma-\mu_0\left(\mathbf{s}_{k+T,k-L-T+1}\right)}{\sigma_0\left(\mathbf{s}_{k+T,k-L-T+1}\right)}\right),
 \end{align}
 where $\mathbf{Q}$ is the Gaussian Q-function. From \eqref{eq:25}, \eqref{eq:26}, \eqref{eq:44}, and \eqref{eq:43}, we obtain the approximate BER for MC with linear equalization.

 \subsection{Decision Feedback Sequence Detection} \label{sub4:3}

 The complexity of the VA depends on the  channel memory. For MLSD, the number of states of the VA for binary modulation is $2^{L-1}$. The complexity can be reduced by reducing the number of states, i.e., shortening the channel memory. We adopt DFSD \cite{Due:89:TC} that uses the VA with a reduced number of states  $2^{\lambda-1}$, $\lambda\leq L$. For metric calculation, the first $\lambda-1$ symbols, $[s_{k-1},s_{k-2},\dots, s_{k-\lambda+1}]$, are defined by a state and the remaining symbols, $[s_{k-\lambda}, s_{k-\lambda-1}, \dots, s_{k-L+1}]$, are taken from the surviving path of that state, i.e., $[\hat{s}_{k-\lambda}, \hat{s}_{k-\lambda-1}, \dots, \hat{s}_{k-L+1}]$, see Fig.~\ref{fig:4}. Note that $[\hat{s}_{k-\lambda}, \hat{s}_{k-\lambda-1}, \dots, \hat{s}_{k-L+1}]$ can be different for each of the   $2^{\lambda-1}$ states. For the VA for DFSD, the accumulated path metrics $PM_k(\mathbf{s}_{k,k-\lambda+2})$ of  state $\mathbf{s}_{k,k-\lambda+2}$  is given by
 \begin{align} \label{eq:40}
 &PM_k\left(\mathbf{s}_{k,k-\lambda+2}\right)=PM_{k-1}\left(\mathbf{s}_{k-1,k-\lambda+1}\right)\\\nonumber
 &\hspace{1.7cm}+\sum_{m=1}^M \ln p_{g_{k,m}|s_k, \dots, s_{k-\lambda+1}, \hat{s}_{k-\lambda}, \dots, \hat{s}_{k-L+1}}(g_{k,m}).
 \end{align}

\section{Numerical Results}     \label{sec:5}

In this section, we illustrate the performance of the proposed equalizers and detectors in terms of the BER. For   linear equalization, the results are obtained by both analysis and simulation. The results of the nonlinear equalizers and sequence detectors are obtained by  simulation only since  the  BERs of these schemes cannot be derived in  closed-form.   
  We set $\Delta t=T_\mathrm{b}/M$ and $\Tb$ is normalized by $t_\mathrm{peak}$ as $\Tb=\beta t_\mathrm{peak}$, where $t_\mathrm{peak}$ is the time when the expected number of molecules  at the receiver peaks \cite[Eq.~(4), Eq.~(6)]{NCS:14:GLOBECOM}. We use the following system parameters $D=\SI[per-mode=symbol]{4.3e-10}{\meter \per \second^2}$ ,
 $r=\SI{5e-7}{\meter}$,
  $V=\frac{4}{3}\pi\left(5\times 10^{-8}\right)^3 \si{\meter^3}$,
  $v_1=v_2=\SI[per-mode=symbol]
{3e-3}{\meter\per\second}$,
  $M=3$,
  $L=5$,
  $\beta=1.5$,
  $T=1$,
  $L_1=1$,
  $L_2=3$, and
  $\lambda=2$. Note that $L$ is chosen so that for the last tap, the channel impulse response (CIR) $h_{l,m}$ has decayed to  a value of $0.089$, i.e., equal to $0.4\%$  of the value of $h_{l,m}$ at $t_\mathrm{peak}$. Moreover,  the value of $M$ cannot be arbitrarily large  for the independence between the observations to hold.  We simulate $10^7$ transmissions of information bits to obtain the numerical results.

 Fig.~\ref{fig:7} presents the BER as a function of the number of molecules released by the transmitter for the proposed schemes and  benchmark schemes. The proposed schemes are the linear fractionally-spaced equalizer (Subsection~\ref{sub3:1}),   fractionally-spaced DFE (Subsection~\ref{sub3:3}), and  DFSD scheme (Subsection~\ref{sub4:3}). The benchmark schemes from the literature are the symbol-rate equalizer \cite{KA:13:JSAC},  matched filter \cite{JAS:17:CL}, and  MLSD using the VA \cite{PRO:01:Book}. The symbol-rate equalizer \cite{KA:13:JSAC}  is a linear equalizer that uses the received signals sampled at  $t_\mathrm{peak}$ in $T$ symbol intervals before and after the  symbol to be detected. The matched filter uses $M$ samples taken in one symbol interval for one symbol detection.  As expected, in Fig.~\ref{fig:7}, the BER decreases with increasing number of released molecules.
 With the proposed linear fractionally-spaced equalizer, the BER is reduced significantly compared to the BER when using the matched filter or the symbol-rate equalizer. This demonstrates that fractionally-spaced equalization is needed to effectively mitigate ISI in MC. Interestingly, the fractionally-spaced  DFE yields a BER that is not much lower than the BER for the linear fractionally-spaced equalizer. This is because the CIR is favorable for the linear equalizer, i.e., the Z-transform of the CIR does not have zeros close to the unit circle limiting the noise enhancement caused by linear equalization. Moreover,  DFE causes error propagation which eliminates part of the gain it achieves over linear equalization.   DFSD with $2$ states reduces error propagation and can  reduce the BER compared to fractionally-spaced DFE and closely approaches the performance of  MLSD with $2^4=16$ states. 
Moreover, for the linear equalizers, we observe that the analytical results obtained by approximating the Poisson distribution by a Gaussian distribution match well with the simulation results.

\section{Conclusions}

In this paper, we designed fractionally-spaced equalization and detection techniques for diffusive MC, namely the linear fractionally-spaced equalizer, fractionally-spaced DFE,  and DFSD. Due to signal-dependent noise in  diffusive MC channels, the designs of the equalizers and detectors for diffusive MC are different from those  for conventional wireless communication systems.   Our results reveal that significant reductions of the BER are achieved with fractionally-spaced equalizers compared to the symbol-rate equalizer and  matched filter receivers from the  MC literature.  Albeit having a much lower complexity, the proposed DFSD is able to achieve
a similar BER as  MLSD.

	\begin{figure}[!t]
		%		\begin{minipage}[!t]{0.49\textwidth}\hspace*{-5 mm}
		\centering
		%			\resizebox{0.95\linewidth}{!}{
		\includegraphics[scale=0.7]{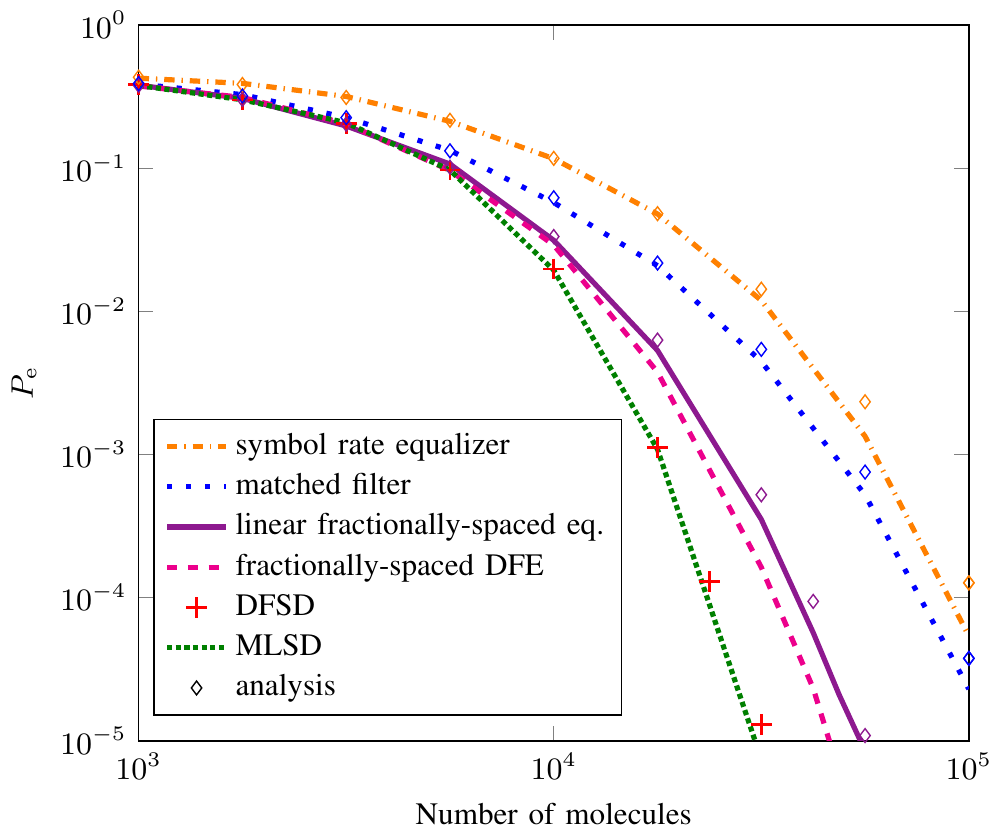}
		\caption{
			BER as a function of the number of molecules for the proposed  and benchmark schemes. 
		}
		\label{fig:7}
		%		\end{minipage}
	\end{figure}

\bibliographystyle{IEEEtran}
\bibliography{IEEEabrv,MolecularBib}

\end{document}